\documentclass[conference]{IEEEtran}
\IEEEoverridecommandlockouts
\usepackage{cite}

\usepackage{amsmath,amssymb,amsfonts,mathtools,bm}
\usepackage{amsthm}
\usepackage{algorithm}
\usepackage{algorithmic}
\usepackage{multicol}
\usepackage{graphicx}
\usepackage{textcomp}
\usepackage{xcolor,float}
\usepackage{tabularx}
\usepackage{textcomp}
\usepackage{diagbox}
\usepackage{svg}
\usepackage{setspace}
\usepackage{booktabs}
\usepackage{subfigure}
\usepackage{changes}
\newtheorem{theorem}{Theorem}
\newtheorem{problem}{Problem}
\newtheorem{lemma}{Lemma}

\usepackage[hidelinks,bookmarks=false]{hyperref}

\allowdisplaybreaks
\renewcommand{\baselinestretch}{1}

\begin{document}

\title{On-Demand Multimedia Delivery in 6G: An Optimal-Cost Steiner Tree Approach}
\author{
\IEEEauthorblockN{
Zien Wang\IEEEauthorrefmark{1},
Xiucheng Wang\IEEEauthorrefmark{1},
Nan Cheng\IEEEauthorrefmark{1},
Wenchao Xu\IEEEauthorrefmark{2},
Wei Quan\IEEEauthorrefmark{3},
Ruijin Sun\IEEEauthorrefmark{1},
Conghao Zhou\IEEEauthorrefmark{1},
}
\IEEEauthorblockA{
\IEEEauthorrefmark{1}School of Telecommunications Engineering, Xidian University, Xi'an, 710071, China\\
\IEEEauthorrefmark{2}Division of Integrative Systems and Design, The Hong Kong University of Science and Technology, 999077, Hong Kong SAR\\
\IEEEauthorrefmark{3}School of Electronic and Information Engineering, Beijing Jiaotong University, Beijing, 100044, China\\
Email: \{zewang\_1, xcwang\_1\}@stu.xidian.edu.cn, \{dr.nan.cheng, dr.wei.quan, conghao.zhou\}@ieee.org,\\ wenchaoxu@ust.hk, sunruijin@xidian.edu.cn
}
}
    
    \maketitle

\IEEEdisplaynontitleabstractindextext

\IEEEpeerreviewmaketitle

\begin{abstract}
The exponential growth of multimedia data traffic in 6G networks poses unprecedented challenges for immersive communication, where ultra-high-definition, multi-quality streaming must be delivered on demand while minimizing network operational costs. Traditional routing approaches, such as shortest-path algorithms, fail to optimize flow multiplexing across multiple destinations, while conventional Steiner tree methods cannot accommodate heterogeneous quality-of-service (QoS) requirements—a critical need for 6G’s personalized services. In this paper, we address a fundamental but unsolved challenge: the minimum flow problem (MFP) with multi-destination, heterogeneous outflow demands, which is pivotal for efficient multimedia distribution such as adaptive-resolution video streaming. To overcome the limitations of existing methods, we propose a two-stage dynamic programming-enhanced On-demand Steiner Tree (OST) algorithm, the first approach that jointly optimizes flow aggregation and QoS-aware path selection for arbitrary outflow requirements. We rigorously prove the optimality of OST using mathematical induction, demonstrating that it guarantees the minimum-cost multicast flow under differentiated service constraints. Extensive experiments in 6G-like multimedia transmission scenarios show that OST reduces total network flow by over 10\% compared to state-of-the-art methods while ensuring on-demand QoS fulfillment. The complete code is available at \texttt{https://github.com/UNIC-Lab/OST}.
\end{abstract}

\begin{IEEEkeywords}
minimum flow problem, on-demand, dynamic programming, Steiner tree.
\end{IEEEkeywords}

\section{Introduction}
The emergence of 6G networks is driving a paradigm shift in multimedia services, enabling bandwidth-intensive immersive applications such as holographic telepresence (requiring 1-10 Tbps/km² traffic density\cite{Menglan}), real-time volumetric video (with 100+ Gbps peak rates)\cite{DBLP:journals/corr/abs-1906-00741}, and multi-sensory extended reality (XR)\cite{Shen,sun2025comprehensive}, which further amplifies the complexity of QoS guarantees in mobile environments \cite{Stafidas2024}. This evolution is accompanied by an explosive growth in global multimedia traffic, with the live streaming market alone projected to grow from \$87.3 billion in 2023 to \$3.718 trillion by 2030 \cite{latreche2025applicationsenvisagednewgeneration}. However, these transformative services impose unprecedented challenges on network infrastructure for personality and economy\cite{radiodiff,wang2025radiodiff}. The requirement to support heterogeneous quality-of-service (QoS) demands where multiple users may request drastically different data resolutions such as , 8K vs. 360p from the same multicast session\cite{Mehran}. Meanwhile, the critical operational constraint of minimizing network flow to reduce energy consumption and infrastructure costs in increasingly complex 6G topologies\cite{Chukhno}. Traditional routing approaches fundamentally fail to address this trilemma of throughput scalability, QoS personalization, and cost efficiency. Shortest-path algorithms, while providing baseline connectivity, incur up to 47\% redundant flowin multi-destination scenarios due to their inability to exploit flow multiplexing \cite{Li2006/01/01}. Although conventional Steiner tree methods are theoretically optimal for homogeneous multicast, they become severely suboptimal\cite{Risso} when confronted with heterogeneous outflow requirements typical of 6G multimedia services. This limitation represents a critical bottleneck for next-generation networks, where the capability to deliver on-demand, quality-adaptive content with minimum resource expenditure will directly determine the economic viability of immersive services\cite{Bhattacharya}.

To bridge the fundamental gap between personalized QoS delivery and network efficiency in 6G multimedia distribution, we propose a novel On-demand Steiner Tree (OST) algorithm that fundamentally rethinks multicast flow optimization by simultaneously addressing flow minimization and quality adaptation through three key innovations. The OST framework introduces an adaptive tree construction mechanism employing a two-stage dynamic programming approach: first, decomposing the network into QoS-homogeneous subgraphs via outflow requirement clustering, then optimally connecting these subgraphs through minimum-cost Steiner junctions. Moreover, a mathematical proof is given that the proposed OST can obtain the optimal transmission flow with minimum cost and 100\% satisfaction of heterogeneous resolution requests. The main contributions of this paper are summarized as follows.
\begin{enumerate}
    \item To our best knowledge, we establish the first formal framework for multicast flow optimization under heterogeneous QoS constraints in 6G networks, introducing the Outflow-Constrained Minimum Flow Problem (OCMFP) as a novel mathematical formulation that captures the critical trade-off between personalized service delivery and network efficiency.
    \item Based on the Steiner tree method an on-demand flow minimization method, named OST, is proposed, which is a dynamic programming-enhanced solution that not only guarantees provably minimum flow (as established through rigorous mathematical induction) but also achieves unprecedented adaptability to diverse user requirements through its two-stage QoS-aware path construction mechanism.
    \item Experiment results demonstrate that our proposed OST method reduces network traffic by over 10\% compared to traditional flow optimization algorithms.
\end{enumerate}

\section{Outflow Constrained Minimum Flow Problem}
In this paper, we address a specialized class of the minimum flow problem, where a single source node transmits flow to multiple destination nodes with varying outflow requirements, denoted by $x$, within a network. This problem is modeled as an undirected weighted graph $\mathcal{G} = \{\mathcal{V}, \mathcal{E}\}$, where $\mathcal{V}$ represents the set of all nodes in the network. The node set is divided into source nodes $\mathcal{V}_{s}$, destination nodes $\mathcal{V}_{d}$, and general routing nodes $\mathcal{V}_{r}$, such that $\mathcal{V} = \{\mathcal{V}_{s}, \mathcal{V}_{d}, \mathcal{V}_{r}\}$, with cardinalities $|\mathcal{V}| = M$, $|\mathcal{V}_{d}| = K$, and $|\mathcal{V}_{s}| = 1$. The flow from node $i$ to node $j$ is denoted by $f_{(i,j)}$, and the weight of the link $e_{(i,j)}$ represents the unit cost of transmitting flow from node $i$ to node $j$.

Therefore, the OCMFP can be formulated as follows.
\begin{problem}\label{p1}
    \begin{align}
        &\min_{\bm{f}} \sum_{i\in \mathcal{V}}\sum_{j\in \mathcal{V}} e_{(i,j)}f_{(i,j)},\label{obj}\\
        &s.t.\;  \max_{j\in\mathcal{V}} f_{(i,j)} \leq \max_{k\in\mathcal{V}} f_{(i,j)}, &&\forall i\in\mathcal{V}/\mathcal{V}_{s},\tag{\ref{obj}a}\label{c1}\\
        &\quad\;\; \max_{j\in \mathcal{V}} f_{(i,j)}\geq \max_{k} x^{k},&&\forall k\in\mathcal{V}_{d} \wedge i \in\mathcal{V}_{s},\tag{\ref{obj}b}\label{c2}\\
        &\quad\;\;x^{j}\leq \max_{i\in \mathcal{V}\setminus\mathcal{V}_{d}} f_{(i,j)}, &&\forall j\in\mathcal{V}_{d} ,\tag{\ref{obj}c}\label{c3}\\
        &\quad\;\; f_{(i,j)} \geq 0, &&\forall i,j \in \mathcal{V},\tag{\ref{obj}d}\label{c4}\\
        &\quad\;\; f_{(i,j)} = 0, &&\forall e_{(i,j)}\notin \mathcal{E},\tag{\ref{obj}e}\label{c5}
    \end{align}
\end{problem}
\noindent The objective of Problem \ref{p1} is to minimize the weight summation flow within the network. Constraint \eqref{c1} ensures that for each node in the graph, except for the source node, the outflow does not exceed the inflow. Constraint \eqref{c2} guarantees that the outflow from the source node exceeds the maximum inflow requirements of the destination nodes, thereby satisfying all inflow demands. Constraint \eqref{c3} ensures that the inflow at each destination node is at least equal to its specified inflow requirement. Constraints \eqref{c4} and \eqref{c5} ensure non-negative flow within the network and restrict flow to the links that exist within the graph. The optimal solution to Problem \ref{p1} adheres to the following theorem.
\begin{theorem}\label{theorem-1}
    For links carrying flow, i.e., where $f_{(i,j)} > 0$ and $e_{(i,j)} \in \mathcal{E}$, these links form a tree structure with the source node as the root and all destination nodes as the leaf nodes. The flow on each link follows the direction from the root node to the leaf nodes, corresponding to increasing depth in the tree. The minimum flow that satisfies the outflow requirements of all destination nodes must conform to this tree structure.
\end{theorem}
\begin{proof}
    The tree structure described in Theorem \ref{theorem-1} obviously satisfies constraints \eqref{c2}, \eqref{c4}, and \eqref{c5}. To satisfy constraints \eqref{c1} and \eqref{c3}, there must be a path between the source node and all destination nodes. The flow-carrying links and their corresponding nodes can form some sub-graphs of the original graph, if there is a disconnected sub-graph, then any connected component of this sub-graph that is not connected to the source node cannot be connected to all destination nodes. By eliminating the flows in this isolated component, the summation flow can be reduced while guaranteeing the constraints. Assuming that there is a sub-graph that is not a tree, if a cycle exists within the sub-graph, let $v$ be the node in the cycle closest to the source. Removing the flow-carrying link on the cycle that flows to $v$ still satisfies the constraint conditions, yielding a better solution. Similarly, if there are distinct nodes $u$ and $v$ connected by multiple disjoint paths, removing any flow-carrying link along one of these paths will still satisfy the constraint conditions, resulting in an improved solution. Finally, if the tree contains a leaf node $v$ that is not a destination node, the path from its deepest ancestor with the largest number of children to $v$ can be deleted without violating the constraint conditions. This further optimizes the solution. Thus, the optimal solution to Problem \ref{p1} must take the form of the tree structure specified in Theorem \ref{theorem-1}.
\end{proof}

\begin{lemma}
    In the optimal solution of Problem \ref{p1}, if a link is multiplexly used to transmit the flow from the source node to multiple destination nodes, the flow size on this link equals the highest outflow demand among those destination nodes.
\end{lemma}
\begin{proof}
    If the flow size of this link is smaller than the highest outflow demand among those destination nodes, then according to Theorem \ref{theorem-1} and constraint \eqref{c1}, the size of the flow passing to the destination node with the highest outflow requirements cannot meet its outflow demand.
\end{proof}

\section{On-demand Steiner Tree Method}
\begin{figure*}
    \centering
    \includegraphics[width=0.965\linewidth]{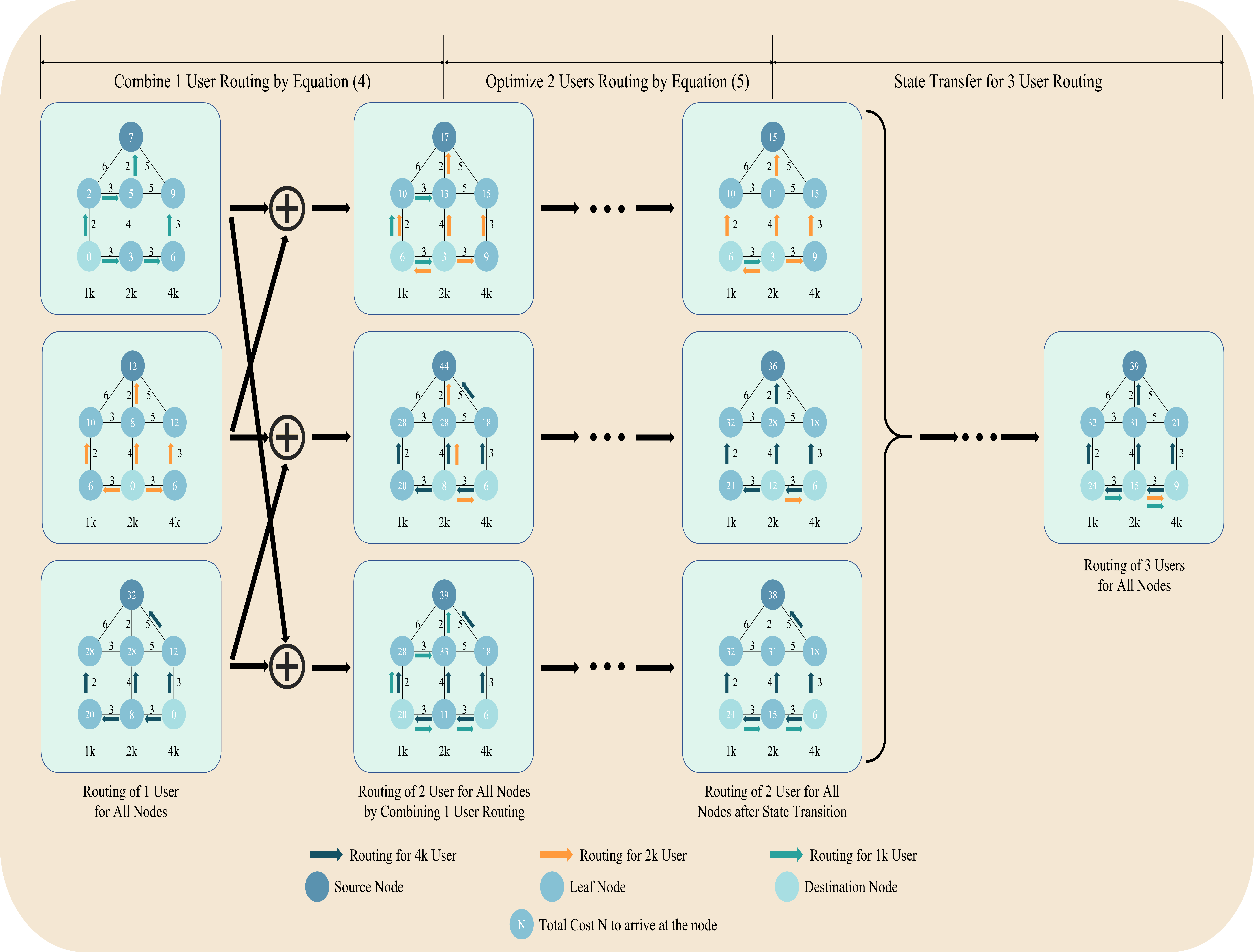}
    \caption{{Application of OST algorithm.}}
    \label{fig-system}
\end{figure*}
In order to resolve OCMFP effectively, the proposed OST is based on dynamic programming. OST obeys the following state transition equations.  
\begin{equation}
    \mathcal{H}_{(d,\{d\})} = 0 \quad d \in \mathcal{V}_d
    \label{e1}
\end{equation}
\begin{equation}
    \mathcal{H}_{(v,\emptyset)} = +\infty \quad v \in \mathcal{V}
    \label{e2}
\end{equation}
\begin{equation}
    \mathcal{H}_{(v,\mathcal{S})} = \min_{\mathcal{F} \subseteq \mathcal{S}} \left\{ \max \left( \mathcal{H}_{(v,\mathcal{F})}, \mathcal{H}_{(v,\mathcal{S} \setminus \mathcal{F})} \right) \right\} \quad v \in \mathcal{V}, S \subseteq \mathcal{V}_d
    \label{e3}
\end{equation}
\begin{equation}
    \mathcal{H}_{(v,\mathcal{S})} = \min_{(v,u) \in \mathcal{E}} \left\{ \mathcal{H}_{(u,\mathcal{S})} + \max_{j \in S} \{ x^j \}e_{(v,u)} \right\} \quad v \in \mathcal{V}, S \subseteq \mathcal{V}_d
    \label{e4}
\end{equation}
where $\mathcal{H}_{(v,\mathcal{S})}$ is the total flow from node $v$ to destination nodes set $\mathcal{S}$. 

Equation \eqref{e1} and equation \eqref{e2} define specific boundary essential for the link selection process. In detail, equation \eqref{e1} indicates that the flow of each source node to itself is zero, ensuring that no additional cost is added when the destination node is reached. Equation \eqref{e2} represents that the flow is set to infinity for a node $v$ when there is no source to transmit to it.

\begin{theorem}
    Set $\delta_{(v,\mathcal{S})}$ as the theoretical minimum flow from node $v$ to destination nodes set $\mathcal{S}$. Through the above four state transition functions, the $\mathcal{H}_{(v,\mathcal{S})}=\delta_{(v,\mathcal{S})}$ can be obtained.
\end{theorem}
\begin{proof}
    When $|\mathcal{V}_d|=1$, which means that there is only one destination node. Set $\mathcal{V}_d = \{a\}$. Then The problem then degenerates into a single-source shortest path problem on $\mathcal{G}$. The transition described in Equation \eqref{e3} becomes irrelevant because either $\{a\}\setminus\mathcal{F}$ or $\mathcal{F}$ itself will inevitably include an empty set. Equation \eqref{e4} consequently recduces to the standard update formula for shortest paths. It can be easily proven that within $\mathcal{G}$, $\mathcal{H}_{(v,\mathcal{S})} = \delta_{(v,\mathcal{S})}$ holds true.

    Assume that the conclusion holds when $|\mathcal{V}_d|=k$. Then, when $|\mathcal{V}_d|=k+1$, for any $\mathcal{F}\subseteq\mathcal{V}_d$ with $|\mathcal{F}|\leq k$, it can be easily proven that $\delta_{(v,\mathcal{F})}$ remains consistent when T is considered as $\mathcal{V}'_d$ with smaller number of destination nodes, $|\mathcal{V}'_d|$. Therefore, $\mathcal{H}_{(v,\mathcal{F})}=\delta_{(v,\mathcal{F})}$ holds true. Thus, it only remains to prove that $\mathcal{H}_{(v,\mathcal{V}_d)}=\delta_{(v,\mathcal{V}_d)}$.
    
    From Equation \eqref{e3}, we have 
    \begin{equation}
     \delta_{(v,\mathcal{V}_d)}\leq \mathcal{H}_{(v,\mathcal{V}_d)} \leq \min_{\mathcal{F} \subseteq \mathcal{V}_d} \left\{ \max \left( \delta_{(v,\mathcal{F})}, \delta_{(v,\mathcal{V}_d \setminus \mathcal{F})} \right) \right\}
    \label{e5}
    \end{equation}
    and there exists at least one node $v$ for which both inequalities hold as equalities. Denote such a node as $v'$. Starting from $v'$ and updating its adjacent nodes using Equation \eqref{e4}, we obtain 
    \begin{equation}
     \delta_{(v,\mathcal{V}_d)} \leq \mathcal{H}_{(v,\mathcal{V}_d)} \leq \min_{(v',v) \in \mathcal{E}} \left\{ \delta_{(v,\mathcal{V}_d)} + \max_{j \in \mathcal{V}_d} \{ r_j \}e_{(v',v)} \right\}
    \label{e5}
    \end{equation}
    and again, there exists at least one node $v$ for which both inequalities hold as equalities. Denote such a node as $v''$. Repeating this process until $\forall v\in \mathcal{V}$, we obtain $\mathcal{H}_{(v,\mathcal{V}_d)}=\delta_{(v,\mathcal{V}_d)}$, thus proving the conclusion.
\end{proof}
The details of the proposed OST are specified in Algorithm \ref{alg:sh} and Algorithm \ref{alg:video_transmission} where $\mathcal{T}$ is links set with flows. The Shortest Video Transmission Path algorithm is tailored for establishing the most efficient routing paths for video transmission across a network. The procedure initiates at a source node \( u \), iterating through each neighboring node \( v \). For every neighbor, the algorithm seeks to augment the current path \( \mathcal{T}_{(u,\mathcal{V}^{s}_{d})} \) by integrating the edge unit cost \( e^v_u \) and updating the flow \( f^{tmp} \). Should the cumulative flows along this newly proposed path configuration remain below the predefined threshold \( \mathcal{H}_{(v,\mathcal{V}^{s}_{d})} \), both the path and flow are updated to reflect this more optimal arrangement. The function recursively proceeds to each neighboring node \( v \) that does not coincide with the source set \( \mathcal{V}_s \), culminating in the return of the most efficient path configuration accompanied by the corresponding flow rates and associated costs.

\begin{algorithm}[h]
    \caption{Shortest Algorithm for Video Transmission}
    \label{alg:sh}
    \renewcommand{\algorithmicrequire}{\textbf{Function}}
    \begin{algorithmic}[1]
        \REQUIRE ShortestVideo$(u,\mathcal{V}^{s}_{d})$
        % \ENSURE Optimized minimum flow for specific destination set.
        \FOR{$v\in\mathcal{N}(u)$}
            \STATE $\mathcal{T}^{tmp} = \mathcal{T}_{(u,\mathcal{V}^{s}_{d})}\bigcup\{e_{(v,u)}\}$
            \STATE $f^{tmp} = f_{(u,\mathcal{V}^{s}_{d})}$
            \STATE $f^{tmp}_{(v,u)} = x^{max} e_(v,u)$
            % \STATE $c = \sum_{e^j_k \in \mathcal{T}^{tmp}} f^{tmp}_{j,k}$
            \IF{$\sum_{e_(j,k) \in \mathcal{T}^{tmp}} f^{tmp}_{(j,k)}\leq \mathcal{H}_{(v,\mathcal{V}^{s}_{d})}$}
                \STATE $\mathcal{H}_{(v,\mathcal{V}^{s}_{d})} = \sum_{e_(j,k) \in \mathcal{T}^{tmp}} f^{tmp}_{(j,k)}$;
                \STATE $\mathcal{T}_{(v,\mathcal{V}^{s}_{d})} = \mathcal{T}^{tmp}$
                \STATE $f_{(v,\mathcal{V}^{s}_{d})} = f^{tmp}$
            \ENDIF
            \IF{$v \neq \mathcal{V}_s$}
                \STATE $\mathcal{T}$, $f$, $\mathcal{H} = $ ShortestVideo$(v,\mathcal{V}^{s}_{d})$
            \ENDIF
        \ENDFOR       
        %\STATE Using the shortest algorithm to update the follows: $\mathcal{H}[i\  for\  i\in\mathcal{V},h_{1}]$ and $\mathcal{T}[i\  for\  i\in\mathcal{V},h_{1}]$;
        \RETURN $\mathcal{T}$, $f$, $\mathcal{H}$;
    \end{algorithmic}
\end{algorithm}

%\subsection{Optimize Video Transmission Paths Algorithm}
The \textit{Optimize Video Transmission Paths} algorithm adopts a dynamic programming methodology to optimize video routing efficiency within a network by initializing the path array \( \mathcal{T} \), flow \( f \), and cost \( \mathcal{H} \) to default states, and setting the transmission cost to zero for each destination node \( v \). Subsequently, the algorithm systematically examines every subset \( \mathcal{V}^{s}_{d} \) of the destination nodes, prioritizing subsets based on their cardinality. For each subset, all feasible non-empty combinations \( \hat{\mathcal{J}} \) are evaluated in order of increasing size. Each combination undergoes an assessment involving path merging and computation of the cumulative flow \( f^{tmp} \) and cost \( c \). If this calculated cost \( c \) is less than the currently recorded minimum, the corresponding path and cost values are updated. Throughout this iterative refinement, each node \( v \) is recursively explored using the \textit{Shortest Video} function, guaranteeing an exhaustive examination of all feasible routing paths before finalizing the optimal configurations, associated flows, and costs from source nodes \( \mathcal{V}_s \) to destination nodes \( \mathcal{V}_d \).

Both the time and space complexities of the Optimize Video Transmission Paths algorithm scale exponentially in the size of the target node set $\mathcal{V}_d$. The primary computational burden of the Optimize Video Transmission Paths algorithm arises from the enumeration and processing of non-empty subsets of the target node set, $\mathcal{V}_d$. Initially, all non-empty subsets of $\mathcal{V}_d$ are considered, whose total number is $2^{|\mathcal{V}_d|} - 1$. For each such subset $\mathcal{V}^{s}_{d}$, the algorithm further enumerates its non-empty sub-subsets, which in the worst case can also be on the order of $2^{|\mathcal{V}_d|}$ iterations. Within these nested loops, for every node $v \in \mathcal{V}$ the algorithm performs state updates and invokes the \textsc{ShortestVideo} routine; assuming the latter operates in $O(|\mathcal{E}|)$ time in a sparse graph, the overall per-subset cost becomes $O(|\mathcal{V}| \cdot 2^{|\mathcal{V}_d|} + |\mathcal{V}| \cdot |\mathcal{E}|)$. Aggregating over all $2^{|\mathcal{V}_d|}$ subsets, the worst-case time complexity of the algorithm can be summarized as \[ O\Bigl(|\mathcal{V}| \cdot 2^{2|\mathcal{V}_d|} + |\mathcal{V}| \cdot |\mathcal{E}| \cdot 2^{|\mathcal{V}_d|}\Bigr), \] showing an exponential growth in time with respect to the number of target nodes.

In terms of space, the algorithm must maintain optimal state information for each node $v \in \mathcal{V}$ with respect to every non-empty subset of target nodes $\mathcal{V}^{s}_{d} \subseteq \mathcal{V}_d$, which is $\mathcal{T}_{(v,\mathcal{V}^{s}_{d})}$. This involves storing path sets, flow variables, and cumulative metrics for each combination. Since there are up to $2^{|\mathcal{V}_d|} - 1$ such subsets, the required storage per node becomes $O(2^{|\mathcal{V}_d|})$, leading to an overall space complexity of 
\[
O\bigl(|\mathcal{V}| \cdot 2^{|\mathcal{V}_d|}\bigr).
\]

\begin{algorithm}[h]
    \caption{Optimize Video Transmission Paths}
    \label{alg:video_transmission}
    \renewcommand{\algorithmicrequire}{\textbf{Function:}}
    \renewcommand{\algorithmicensure}{\textbf{Output:}}
    \begin{algorithmic}[1]
        \REQUIRE DpVideo()
        \STATE Set each element in $\mathcal{T}$, $f$, and $\mathcal{H}$ as $\emptyset$, $0$, and $+\infty$;
        \FOR{$v\in \textbf{enumerate}(\mathcal{V}_d)$}
            \STATE $\mathcal{H}_{(v,\{v\})} = 0$;
        \ENDFOR
        \STATE $\mathcal{J} = \{ \mathcal{V}^{s}_{d} \subseteq \mathcal{V}_d : |\mathcal{V}^{s}_{d}| > 0 \}$
        \STATE $\text{sort}(\mathcal{J}, \text{by } |\mathcal{V}^{s}_{d}|)$
        \FOR{$\mathcal{V}^{s}_{d} \in \mathcal{J}$}
        % \FOR{$h_1$ from $0$ to $2^{|\mathcal{V}_{d}|} - 1$}
            \STATE $\hat{\mathcal{J}} = \{ \hat{\mathcal{V}^{s}_{d}} \subseteq \mathcal{V}^{s}_{d} : |\hat{\mathcal{V}^{s}_{d}}| > 0 \}$
            \STATE $\text{sort}(\hat{\mathcal{J}}, \text{by } |\hat{\mathcal{V}^{s}_{d}}|)$
            \STATE $x^{max} = \max \{x^v : v \in \mathcal{V}^{s}_{d}\}$
            \FOR{$v\in \mathcal{V}$}
                \FOR{$\hat{\mathcal{V}^{s}_{d}} \subseteq \hat{\mathcal{J}}$}
                    \IF{$\left(\mathcal{H}_{(v,\hat{\mathcal{V}^{s}_{d}})} \,+\, \mathcal{H}_{(v,\mathcal{V}^{s}_{d} \setminus \hat{\mathcal{V}^{s}_{d}})}\right) \neq +\infty$}
                        \STATE $\mathcal{T}_{tmp} = \mathcal{T}_{(v, \hat{\mathcal{V}^{s}_{d}})} \cup \mathcal{T}_{(v, \mathcal{V}^{s}_{d} \setminus \hat{\mathcal{V}^{s}_{d}})}$;

                        \FOR{$e^m_n \in \mathcal{T}^{tmp}$}
                            \STATE $f^{tmp}_{(m,n)} = \max(f_{(m,n,v, \hat{\mathcal{V}^{s}_{d}})}, f_{(m,n,v, \mathcal{V}^{s}_{d} \setminus \hat{\mathcal{V}^{s}_{d}})})$;
                        \ENDFOR
                        \STATE $c = \sum_{e^{m}_{n} \in \mathcal{T}^{tmp}} f^{tmp}_{(m,n)}$;

                        \IF{$c\leq \mathcal{H}_{(v,\mathcal{V}^{s}_{d})}$}
                            \STATE $\mathcal{H}_{(v,\mathcal{V}^{s}_{d})}=c$;
                            \STATE $\mathcal{T}_{(v,\mathcal{V}^{s}_{d})}=\mathcal{T}^{tmp}$;
                        \ENDIF
                    \ENDIF
                \ENDFOR
            \ENDFOR
            \FOR{$v\in \mathcal{V}$}
                \STATE $\mathcal{T}$, $f$, $\mathcal{H} = $ ShortestVideo$(v,\mathcal{V}^{s}_{d})$;
            \ENDFOR
        \ENDFOR
        \RETURN $\mathcal{T}_{(\mathcal{V}_s,\mathcal{V}_d)}$, $f_{(\mathcal{V}_s,\mathcal{V}_d)}$, $\mathcal{H}_{(\mathcal{V}_s,\mathcal{V}_d)}$. 
    \end{algorithmic}
\end{algorithm}

\section{Simulation Results and Discussion}
\subsection{Simulation Settings}
In this part, simulation experiments are conducted to evaluate the performance of the proposed OST on OCMFP. The links' unit costs are randomly with $e_{(i,j)} \sim U(0,1)$, and the required video resolution ratio of each destination node is defined with $x^i \sim \text{Discrete}\left(\{1, 0.5, 0.25\}, \left\{\frac{1}{3}, \frac{1}{3}, \frac{1}{3}\right\}\right)
$ which representing $\{4k,2k,1k\}$ of resolution ratio, desperately. We consider the following benchmarks for comparison.
\begin{itemize}
    \item MST: Minimum Spanning Tree(MST) is crucial in combinatorial optimization, connecting all vertices of a weighted graph with the minimum sum of edge weights.\cite{mst}
    \item DA: Dijkstra's Algorithm(DA) is a powerful and reliable method for solving the single-source shortest path problem in graphs with non-negative weights.
    \item GA: Genetic Algorithm (GA) is a type of optimization algorithm inspired by the principles of natural selection and genetics. It is used for solving complex optimization and search problems by mimicking the process of biological evolution.\cite{10293821}
    \item ACO: Ant Colony Optimization (ACO) is a nature-inspired metaheuristic algorithm developed for solving complex optimization problems.
    \item BCO: Bee Colony Optimization (BCO) is a type of swarm intelligence algorithm inspired by the foraging behavior of honeybee swarms and is used for various optimization problems, including combinatorial optimization, function optimization, scheduling, and data clustering.
\end{itemize}
In the simulation, all the algorithms are set to use dynamic pruning during the merging process except MST which just use the maximum requirement for each link flow.
\subsection{Performance Evaluation on Scalability}

In Figure \ref{fig:node} and Figure \ref{fig:smallnode}, comparisons between different algorithms across varying network sizes highlight their scalability and efficiency. The OST algorithm consistently exhibits the lowest transmission costs, demonstrating notable scalability and effectiveness, particularly in larger networks where efficient resource management becomes crucial. OST's performance advantage grows more pronounced as network size increases, emphasizing its suitability for extensive networks that demand effective cost management. Conversely, although OST maintains superiority in smaller-scale networks, the relative performance gap narrows, indicating that the complexity overhead introduced by alternative algorithms is relatively insignificant at smaller scales, thereby reducing the observed differences in transmission costs.

\begin{figure}[h]
    \centering
    \includegraphics[width=0.7\linewidth]{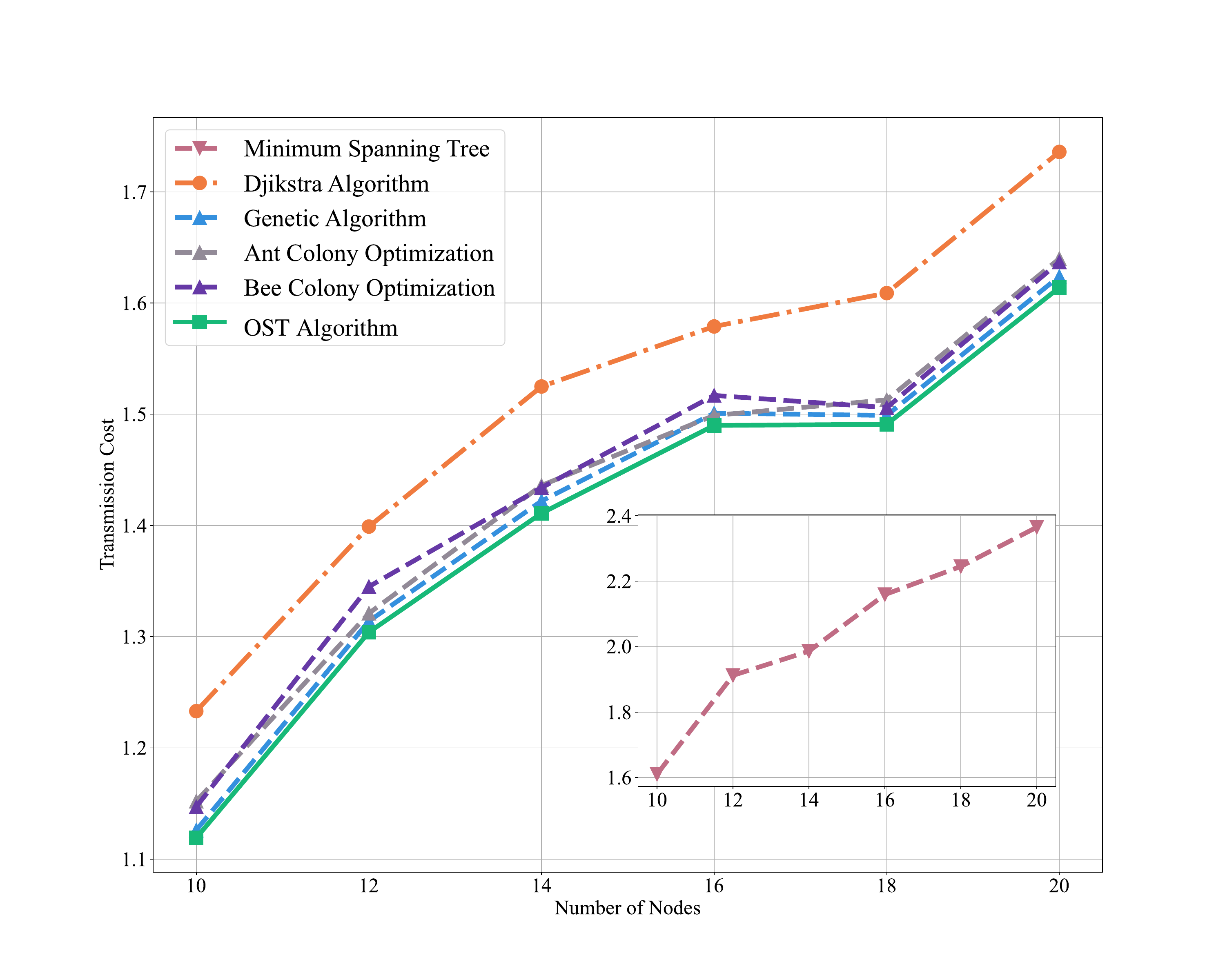}
    \caption{Total Weight Flow on Various Node Number of Different Algorithms}
    \label{fig:node}
\end{figure}

\begin{figure}[h]
    \centering
    \includegraphics[width=0.7\linewidth]{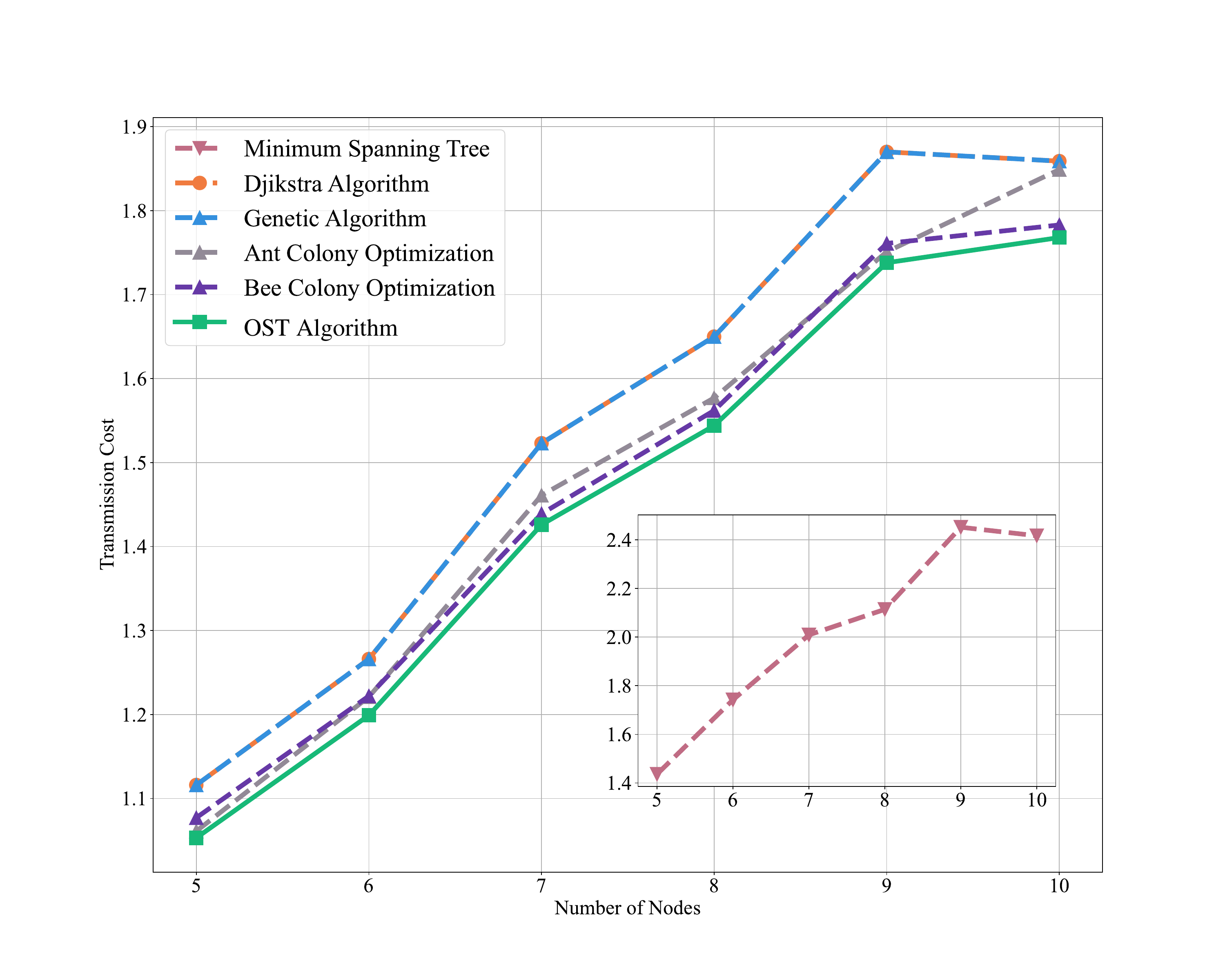}
    \caption{Total Weight Flow on Various Node number of Smaller Size of Different Algorithms}
    \label{fig:smallnode}
\end{figure}

Figures \ref{fig:avrdg} and \ref{fig:degree} illustrate the impact of node connectivity on transmission costs, further reinforcing the advantages of OST. Here, OST consistently maintains lower transmission costs compared to other algorithms as the average degree of connectivity increases. Its robustness in managing networks with higher node connectivity demonstrates superior efficiency, making it particularly beneficial in densely interconnected networks where transmission costs would typically rise sharply. This ability to sustain performance under conditions of high connectivity underlines OST's strength and suitability for complex networking scenarios requiring effective cost optimization.

\begin{figure}[h]
    \centering
    \includegraphics[width=0.7\linewidth]{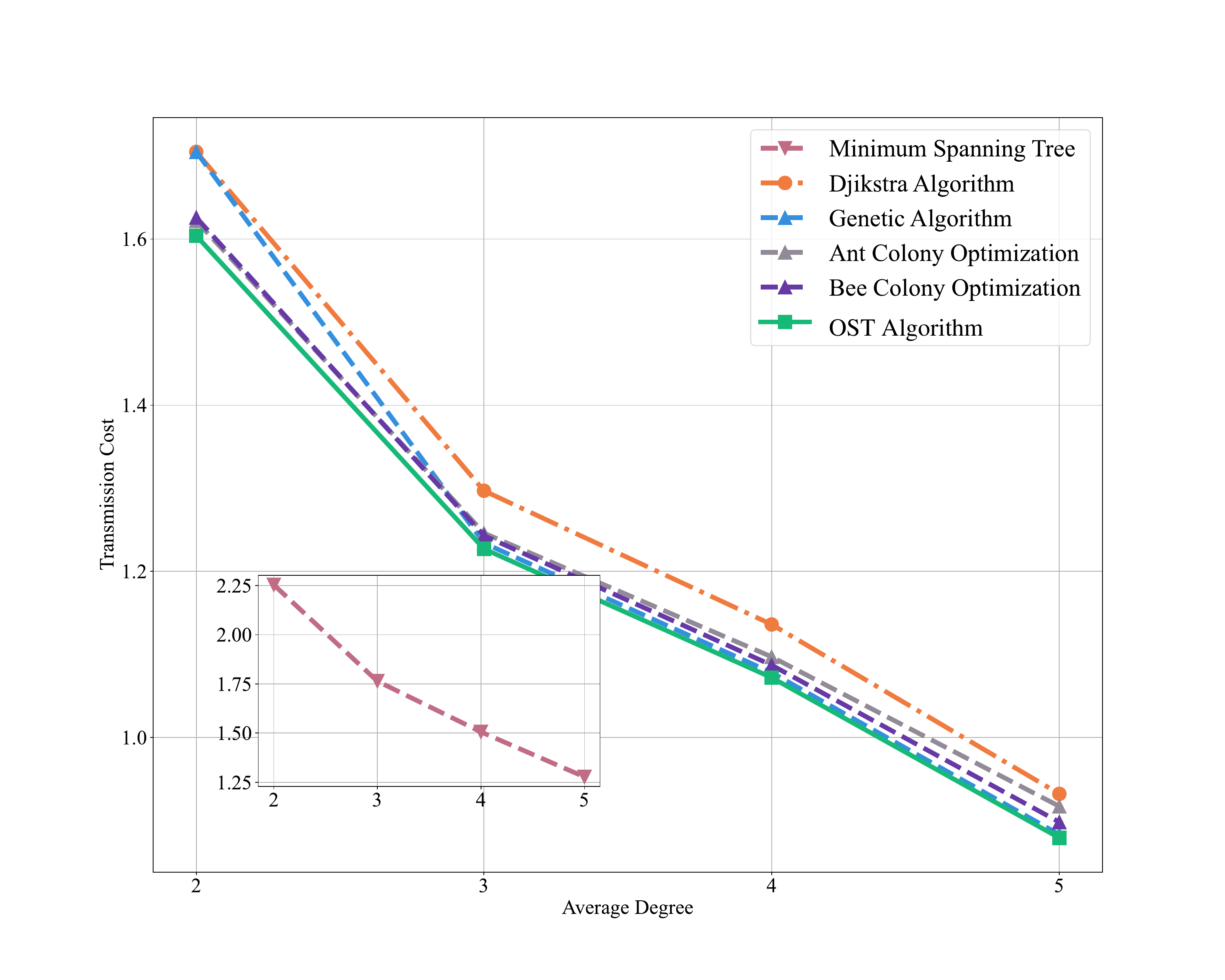}
    \caption{Total Weight Flow on Various Average Degrees of Different Algorithms}
    \label{fig:avrdg}
\end{figure}

\begin{figure}[h]
    \centering
    \includegraphics[width=0.7\linewidth]{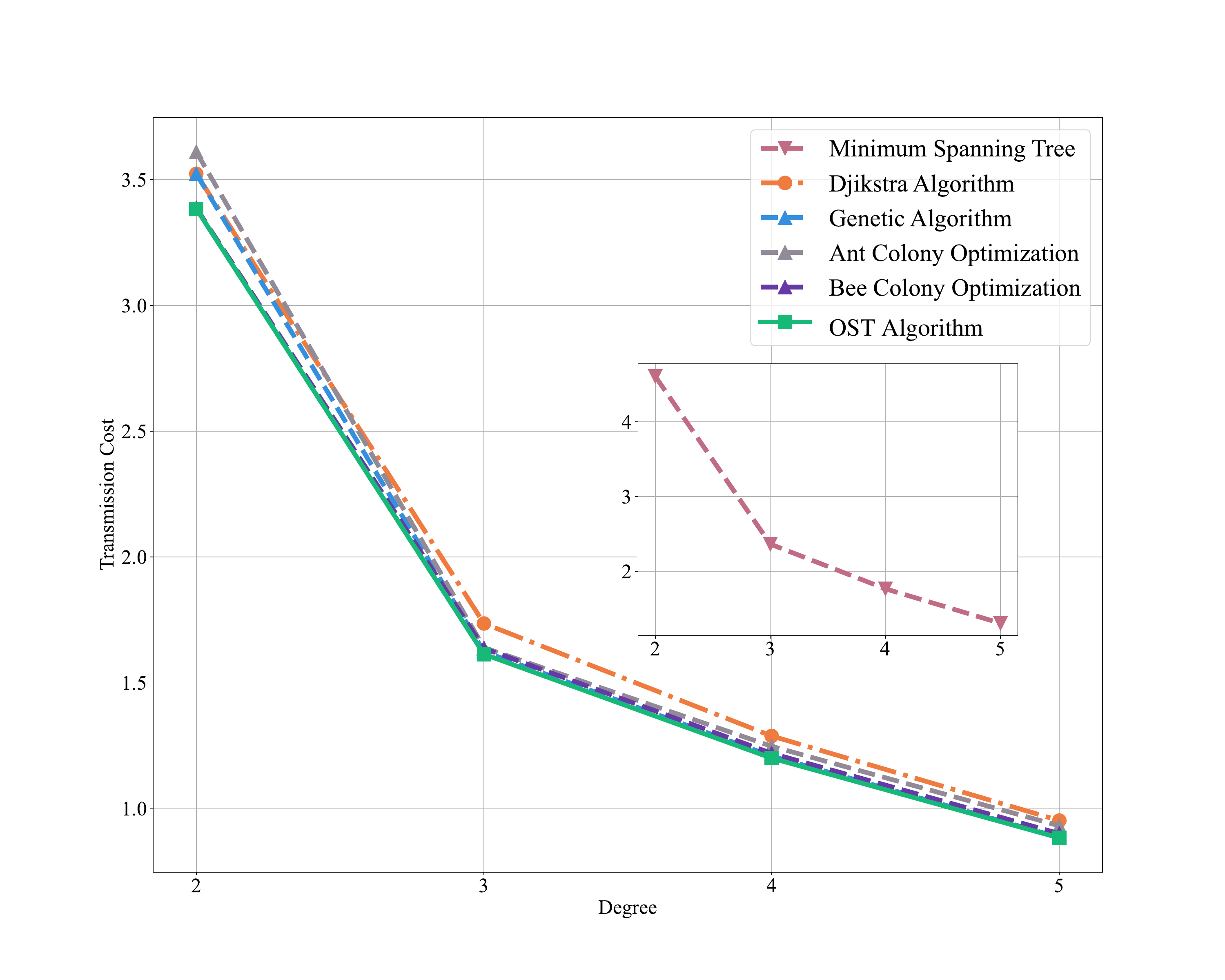}
    \caption{Total Weight Flow on Various Degrees of Different Algorithms}
    \label{fig:degree}
    \vspace{-12pt}
\end{figure}

Figures \ref{fig:user} and \ref{fig:var} further emphasize OST's adaptability and efficiency under various network conditions. As user numbers increase, signifying heightened data loads and concurrent sessions, OST effectively minimizes transmission costs, highlighting its particular effectiveness in managing heavy traffic demands. Additionally, in scenarios characterized by varying network topology and uneven connectivity, OST maintains its superior performance. Its flexibility and adaptability in handling unpredictable network environments enable it to maintain lower transmission costs consistently, distinguishing it from other algorithms that may struggle under such variable conditions.

\begin{figure}[h]
    \centering
    \includegraphics[width=0.7\linewidth]{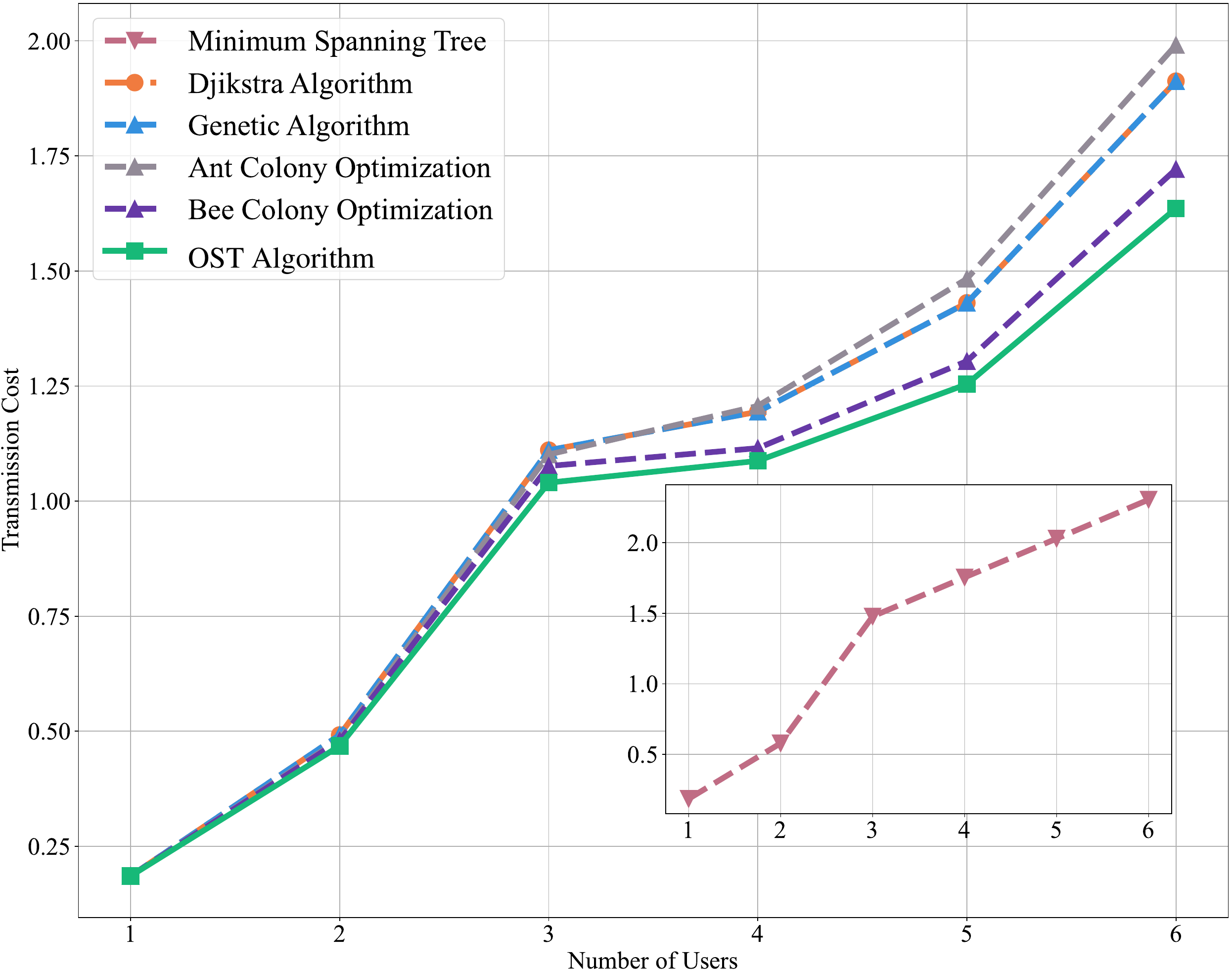}
    \caption{Total Weight Flow on Various User Number of Different Algorithms}
    \label{fig:user}
\end{figure}

\begin{figure}[h]
    \centering
    \includegraphics[width=0.7\linewidth]{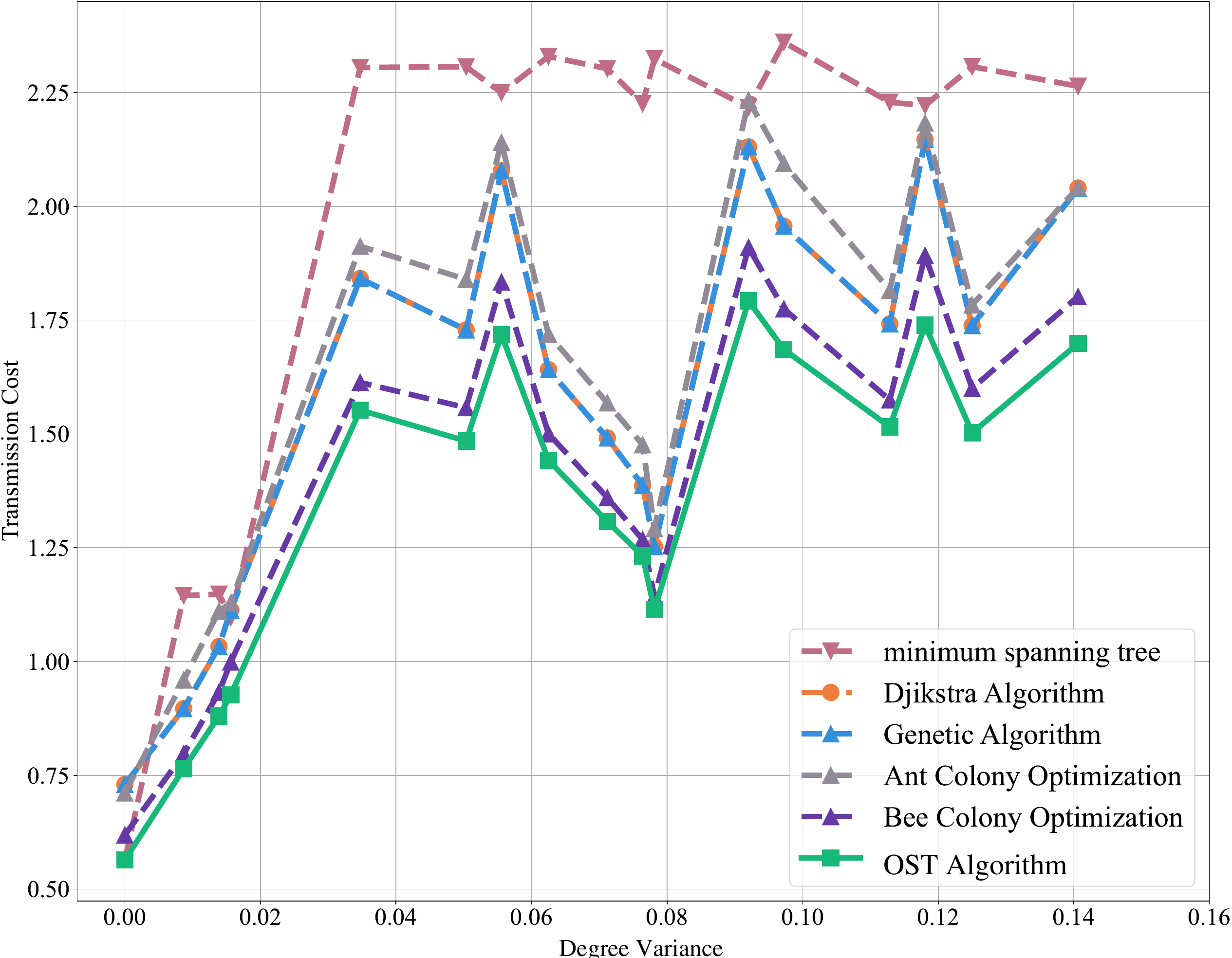}
    \caption{Total Weight Flow on Various Variance of Different Algorithms}
    \label{fig:var}
\end{figure}
Overall, OST outperforms other algorithms across most evaluated scenarios, particularly under conditions of increased network complexity, connectivity, and user load. MST also emerges as a robust solution for large-scale networks, providing efficient and scalable cost management as network size grows. However, OST remains the optimal choice for handling highly interconnected, high-traffic, or variably connected networks due to its inherent adaptability and consistently superior resource management capabilities, making it the preferred algorithm in scenarios demanding dynamic adjustments and optimal cost efficiency.

\section{Conclusion}
This paper presented the On-demand Steiner Tree (OST) method for solving the minimum flow problem in networks with heterogeneous outflow requirements. By combining dynamic programming with Steiner tree optimization, OST overcomes the limitations of traditional approaches that either ignore flow multiplexing or assume homogeneous demands. Experimental results demonstrate OST’s 10\%+ reduction in network flow compared to existing methods while precisely meeting diverse QoS requirements. The algorithm’s robustness across varying network topologies and demand patterns makes it particularly suitable for 6G multimedia delivery. Future work will investigate OST’s integration with SDN controllers and scalability in ultra-dense network deployments.

\section*{Acknowledge}
This work was supported by the National Key Research and Development Program of China (2022YFB2901900).

\bibliography{ref}
\bibliographystyle{IEEEtran}
\ifCLASSOPTIONcaptionsoff
  \newpage
\fi

\end{document}